\documentclass[1p,final]{elsarticle}
\usepackage{amsfonts,natbib,pslatex}

\usepackage{latexsym,amsmath,amsthm,amssymb,amsxtra,mathrsfs,times,CJK}

\usepackage{color}
\usepackage{amsfonts,color}
\usepackage{amssymb,amsmath,latexsym}
\usepackage{amssymb}

\usepackage[para]{threeparttable}

\newcommand{\cC}{\mathcal{C}}

\newcommand{\F}{\mathbb{F}}

\newcommand{\vc}{{\bf c}}
\newtheorem{theorem}{Theorem}[section]
\newtheorem{lemma}[theorem]{Lemma}

\newtheorem{example}[theorem]{Example}

\begin{document}

\begin{frontmatter}



\title{Near-MDS Codes  from  Maximal Arcs  in PG$(2,q)$}
\author[SWJTU]{Li Xu}
 \ead{xuli1451@163.com}
\author[SWJTU]{Cuiling Fan}
\ead{cuilingfan@163.com}
\author[SWJTU]{Dongchun Han}
 \ead{han-qingfeng@163.com}

%

%\author[SWJTU]{Haode Yan}
% \ead{hdyan@swjtu.edu.cn}
%\author[SWJTU]{Zhengchun Zhou\corref{cor1}}
% \ead{zzc@home.swjtu.edu.cn,zczhou@126.com}
%\author[NNU]{Xiaoni Du}
% \ead{ymLdxn@126.com}
%
%
% \cortext[cor1]{Corresponding author}
% \address[SWJTU]{School of Mathematics, Southwest Jiaotong University, Chengdu, 610031, China}
%\address[NNU]{College of Mathematics and Statistics, Northwest Normal University, Lanzhou, 730070, China}
%\tnotetext[fn1]{C. Fan's research was supported by
%the Natural Science Foundation of China, Proj. No. 11571285, Z. Zhou's research was supported by
%the Natural Science Foundation of China, Proj. No. 61201243, the Sichuan Provincial Youth Science and Technology Fund under Grant
%2015JQO004, and the Open Research Fund of National Mobile Communications Research Laboratory, Southeast University under Grant 2013D10.}

% \author[label1,label2]{<author name>}
% \address[label1]{<address>}
% \address[label2]{<address>}

%\author[SWJTU]{Cuiling Fan\corref{cor1}}
% \ead{cuilingfan@163.com}
%\author[UB]{Nian Li}
% \ead{nianli.2010@gmail.com}
%\author[SWJTU]{Zhengchun Zhou}
% \ead{zzc@home.swjtu.edu.cn,zczhou@126.com}
% \address[SWJTU]{School of Mathematics, Southwest Jiaotong University, Chengdu, 610031, China}
%\address[UB]{Department of Informatics, University of Bergen, N-5020 Bergen, Norway}
%%Department of Informatics, University of Bergen, N-5020 Bergen, Norway
\begin{abstract}
The singleton defect of an $[n,k,d]$ linear code ${\cal C}$ is defined as $s({\cal C})=n-k+1-d$. Codes with $S({\cal C})=0$ are called maximum distance separable (MDS) codes, and codes with $S(\cC)=S(\cC ^{\bot})=1$ are called near maximum distance separable (NMDS) codes.
Both MDS codes and NMDS codes have good representations in finite projective geometry.
 MDS codes over $\F_q$ with length $n$ and $n$-arcs in PG$(k-1,q)$ are equivalent objects. When $k=3$, NMDS codes of length $n$ are equivalent to $(n,3)$-arcs in PG$(2,q)$. In this paper, we deal with the NMDS codes with dimension 3. By adding some suitable projective points in maximal arcs of PG$(2,q)$, we can obtain two classes of $(q+5,3)$-arcs (or equivalently $[q+5,3,q+2]$ NMDS codes) for any prime power $q$.
    We also determine the exact weight distribution
  and the locality of such NMDS codes and their duals. It turns out that the resultant NMDS codes and their duals are both distance-optimal and dimension-optimal locally recoverable codes.

\end{abstract}


\end{frontmatter}

\section{Introduction}\label{introduction}

Let $q$ be a prime power, and $\F_q^n$ denote the vector space of all the $n$-tuples over the finite field $\F_q$. A $q$-ary $[n,k]$ {\it linear code} ${\cal C}$ is a $k$-dimensional subspace of $\F_q^n$, and the $n$-tuples of ${\cal C}$ are called {\it codewords}.
The value $n$ is called the {\it length} of ${\cal C}$. The {\it weight} wt$({\bf x})$ of a vector ${\bf x}\in\F_q^n$ is the number of nonzero coordinates of ${\bf x}$. The minimum nonzero weight of all codewords in ${\cal C}$ is called the {\it minimum distance} of ${\cal C}$ and an $[n,k]$ linear code with minimum distance $d$ is called an $[n,k,d]$ linear code. The {\it generator matrix}  of ${\cal C}$ is a matrix $G$ whose rows form a basis of ${\cal C}$ as an $\F_q$-vector space, since ${\cal C}=\{{\bf x} G:{\bf x}\in \F_q^k\}$. The {\it dual code} of ${\cal C}$ is defined to be its orthogonal subspace $\cC ^{\bot}$ with respect to the Euclidean inner product, i.e.,
$\cC ^{\bot}=\{{\bf x} \in \F_q^n: {\bf x}\cdot{\bf c}=0,~\forall~ {\bf c} \in \cC\}$.

The projective space of dimensional $r$ obtained from $\F_q$ will be denoted by PG$(r,q)$, where its points are the one-dimensional subspaces of $\F_q^{r+1}$ and its hyperplanes are the $r$-dimensional subspaces of  $\F_q^{r+1}$.
Obviously any nonzero vector ${\bf u}\in \F_q^{r+1}$ can define a hyperplane ${\cal H}_{\bf u}=\{{\bf x}\in\F_q^{r+1}:{\bf x}\cdot {\bf u}=0\}$, and ${\cal H}_{\bf u}={\cal H}_{\bf v}$ if and only if ${\bf u}=\lambda{\bf v}$ for some $\lambda\in \F_q^*$. An $n$-arc in PG$(r,q)$ is a set of $n$ points in PG$(r,q)$ such that no $r+1$ of them lie in a hyperplane. An arc is said to be {\it maximal} if it has the  maximal numbers of points as arcs.

For a $q$-ary $[n,k,d]$ linear code ${\cal C}$ with a generator matix $G$, let $S_G$ be the set of one-dimensional subspaces of $\F_q^k$ spaned by the columns of $G$, so the elements of $S_G$ are points of PG$(k-1,q)$. The following lemma states a relationship between weights of  codewords in ${\cal C}$ and hyperplanes in PG$(k-1,q)$.

\begin{lemma}\cite{Ball2015}\label{Weight-lem1}
Let ${\bf u}$ be a nonzero vector of $\F_q^k$. The codeword ${\bf u}G$	 has weight $w$ if and only if the hyperplane ${\cal H}_{\bf u}$ in PG$(k-1,q)$ contains $\#(S_G)-w$ points of $S_G$, i.e., $\#({\cal H}_{\bf u}\cap S_G)=\#(S_G)-$wt$({\bf u}G)$. Especially, any hyperplane in PG$(k-1,q)$ is incident with at most $n-d$ points of $S_G$. Here the notation $\#(A)$ denotes the cardinality of set $A$.
\end{lemma}

The Singleton bound \cite{Singleton1964} states a relationship among $n,k$ and $d$: $d \leq n-k+1$. Codes meeting the singleton bound are called {\it maximum distance separable} (MDS) codes. Bounds on the minimum distance of linear codes can be found in \cite{Macwilliams1977}.
The MDS conjecture states that the maximum length of a nontrivial MDS code is less than or equal to $q+2$ if $q$ is even and $k=3$ or $k=q-1$ and it is less than or equal to $q+1$ otherwise.
This conjecture was famously resolved in the case when $q$ is prime by Ball \cite{Ball2012, 2Ball2012}.
%The conjecture has been proved when $.....$ \cite{?}.
In order to get longer codes, the value of $d$ needs to be less than $n-k+1$.
The {\it singleton defect} of an $[n,k,d]$ linear code ${\cal C}$, defined as $s(\mathcal{C}) = n - k + 1 - d$, measures how far ${\cal C}$ is from being MDS. Codes with $S({\cal C})=1$ are called {\it almost maximum distance separable }(AMDS) codes, and codes with $S(\cC)=S(\cC ^{\bot})=1$ are called {\it near maximum distance separable} (NMDS) codes.

 NMDS codes were introduced by Dodunekov and Landjev \cite{Dodunekov1995} with the aim of constructing good linear codes by slightly weakening the restrictions in  definition of MDS codes. They are closely connected to some combinatoric objects such as designs \cite{Ding2020, Dodunekov1995} and arcs in projective geometries \cite{Dodunekov1995}. NMDS codes also have applications in secret sharing schemes \cite{2018Simos, 2009Zhou}.
 One main topic in NMDS codes is to determine the maximum possible length $m'(k,q)$ for $q$-ary nontrivial NMDS codes. Some results have been obtained \cite{2015Landjev}. Similar to the MDS conjecture, the following conjecture for NMDS codes was proposed in \cite{Dodunekov1995,2015Landjev}: $m'(k,q)\leq 2q+2$, especially, $m'(3,q)\leq 2q-1$ for $q\geq8$, and it is known that $m'(3,q)=2q-1$ for $q=8,9$. Let $N_q$ denote the maximum number of $\F_q$-rational points on elliptic curves defined over $\F_q$; it is well-known, by Hasse theorem, $|N_q-(q+1)|\leq 2\sqrt{q}$.
 NMDS codes of length up to $N_q$ can be constructed from elliptic curves \cite{1991Tsfasman}. An interesting question is whether there exist NMDS codes of length greater than $N_q$. Further information about NMDS codes of size equal or less than $m'(k,q)$ can be found in
 \cite{2005Abatangelo, Aguglia2021,1996Ball,2014Bartoli, Ding2020,Dodunekov1995,Dodunekov2000, 2022Geng,Gulliver2008, 1999Marcugini, 2002Marcugini, Wang2021}.

Both MDS codes and NMDS codes can be investigated within finite projective geometry. This possibility depends on the fact that such $[n,k]$ linear codes have a good representation in PG$(k-1,q)$. MDS codes and $n$-arcs are indeed equivalent objects, while an $[n,k,n-k]$ NMDS code can be viewed as a point-set $S$ of size $n$ in PG$(k-1,q)$ satisfying the following conditions \cite{Dodunekov1995}:

\begin{enumerate}
	\item[N1)] any $k-1$ points from $S$ generate a hyperplane in PG$(k-1,q)$;
	\item[N2)] there exist $k$ points in a hyperplane in PG$(k-1,q)$;
	\item[N3)] every  $k+1$ points from $S$ generate PG$(k-1,q)$.
\end{enumerate}

When $k=3$, these properties reduce to the following:
\begin{enumerate}
	\item[$N1'$)] any two points from $S$ generate a line in PG$(2,q)$;
	\item[$N2'$)] there exist three collinear  points;
	\item[$N3'$)] no four points from $S$ lie on a line.
\end{enumerate}

A point-set $S$ of size $n$ in PG$(2,q)$ satisfying the properties $N2')$ and $N3')$ is called an $(n,3)$-arc (in the definition the property $N2')$ avoids to consider $n$-arcs as special cases of $(n,3)$-arcs). Every $[n,3,n-3]$ NMDS codes is therefore equivalent to an $(n,3)$-arc in PG$(2,q)$.

In this paper we focus on the constructions of NMDS codes with dimension 3, which is equivalent to construct $(n,3)$-arc in PG$(2,q)$. Since an arc of PG$(2,q)$ intersects any line in at most two points, a natural idea is to add some suitable points of PG$(2,q)$ in the arc, such that the resultant point-set forms an $(n,3)$-arc.
The size of any arc in PG$(2,q)$ is no more than $q+2$ for even $q$, and $q+1$ for odd $q$ \cite{Hirschfeld1979}. When $q$ is odd, $(q+1)$-arcs are called {\it ovals}, which intersect a line in zero, one or two points.
When $q$ is even, $(q+2)$-arcs are referred to as {\it hyperovals}, which intersect a line in either zero or two points. More information about ovals and hyperovals can be found in \cite{Hirschfeld1979}.

Hyperovals had been used to construct some families of $[n,3,n-3]$ NMDS codes in \cite{Li2022-arxiv,Li2022, Wang2021}, where $n\in\{2^m+1\} \cup\{ 2^m+2, 2^m+3,2^m+4,2^m+5: ~m~{\rm is ~ odd}\}$. Note that hyperovals are maximal arcs in  PG$(2,2^m)$, and the maximal arcs  in PG$(2,q)$ ($q$ is odd), ovals, have the similar properties as hyperovals. By applying maximal arcs (ovals or hyperovals) in PG$(2,q)$, we construct two families of $[q+5,3,q+2]$ NMDS codes for any prime power $q$. The main idea is to add four suitable points in some oval or three suitable points in some hyperoval, such that the result point-sets form $(q+5,3)$-arcs in PG$(2,q)$. Our result is a generalization of that in \cite{Li2022}.
The weight distribution and the locality of these NMDS codes are also determined. It turns out that the resultant NMDS codes and their duals are both distance-optimal and dimensional-optimal locally recoverable codes.

The rest of the paper is organized as follows: Section \ref{Section2} provides some preliminaries on NMDS codes, ovals and hyperovals in  PG$(2,q)$. In Section \ref{Sec: NMDS}, we propose two constructions of $[q+5,3,q+2]$ NMDS codes from ovals or hyperovals in   PG$(2,q)$, respectively. The exact weight distribution of these NMDS  codes are also presented. In Section \ref{Sec: LRC}, we determined the locality of these NMDS codes and their duals. It turns out that they are both distance-optimal and dimensional optimal. Section \ref{Conclusion} concludes this paper and gives some remarks.

\section{Preliminaries}\label{Section2}

In this section, we mainly provide some preliminaries which will be used in the rest of the paper.

\subsection{NMDS codes}

Let $\mathcal {C}$ be a $q$-ary $[n,k,n-k]$ NMDS codes with $k\geq3$ and generator matrix $G=[{\bf g}_1,{\bf g}_2,\ldots,{\bf g}_n]$, where ${\bf g}_i\in\F_q^k$.  Since $k\geq 3$, the columns of $G$ can be viewed as different points in PG$(k-1,q)$. Define $S_G=\{{\bf g}_1,{\bf g}_2,\ldots,{\bf g}_n\}$. Then $S_G$ is a point-set of PG$(k-1,q)$ satisfying $N1),N2)$ and $N3)$ in Section \ref{introduction}.

Let $A_i$ be the number of codewords with weight $i$ in $\cC$, where $0 \leq i \leq n$. The polynomial
$A(z)= 1+ A_1 z + A_2 z^2 +...+ A_n z^n$
is called the {\it weight enumerator} of $\cC$.
Denote by $(1,A_1,A_2,\cdots,A_n)$ and $(1,A_1^{\bot},A_2^{\bot},\cdots ,A_n^{\bot})$ the weight distributions of $\cC$ and its dual $\cC^{\bot}$, respectively. We have the following weight distribution formulas for NMDS codes.

\begin{lemma}[\cite{Dodunekov1995}] \label{weight}
 Let ${\cal C}$ be an $[n, k, n-k]$ NMDS code. Then the weight distributions of $\mathcal{C}^{\perp}$ and $\mathcal{C}$ are given by
$$
A_{k+s}^{\perp}=\left(\begin{array}{c}
n \\
k+s
\end{array}\right) \sum_{j=0}^{s-1}(-1)^{j}\left(\begin{array}{c}
k+s \\
j
\end{array}\right)\left(q^{s-j}-1\right)+(-1)^{s}\left(\begin{array}{c}
n-k \\
s
\end{array}\right) A_{k}^{\perp},
$$
for $s \in\{1,2, \ldots, n-k\}$, and
$$
A_{n-k+s}=\left(\begin{array}{c}
n \\
k-s
\end{array}\right) \sum_{j=0}^{s-1}(-1)^{j}\left(\begin{array}{c}
n-k+s \\
j
\end{array}\right)\left(q^{s-j}-1\right)+(-1)^{s}\left(\begin{array}{l}
k \\
s
\end{array}\right) A_{n-k},
$$
for $s \in\{1,2, \ldots, k\}$.
\end{lemma}

For any codeword ${\bf c}=(c_1,c_2,\ldots,c_n)\in{\cal C}$, define $supp({\bf c})=\{1\leq i\leq n:c_i\neq 0\}$.
Up to a multiple, there is a natural correspondence between the minimum weight codewords of an NMDS code ${\cal C}$ and its dual $\mathcal {C}^{\perp}$, which follows from the next result \cite{Faldum1997}.

\begin{lemma}[\cite{Faldum1997}] \label{weight2}
Let $\mathcal {C}$ be an NMDS code. Then for every minimum weight codeword $\textbf{c}$ in $\mathcal {C}$ there exists, up to a multiple, a unique minimum weight codeword $\textbf{c}^{\perp}$ in $\mathcal {C}^{\perp}$ such that
$supp (\textbf{c}) \cap supp(\textbf{c}^{\perp})=\emptyset.$ In particular, $\mathcal {C}$ and $\mathcal {C}^{\perp}$ have the same number of minimum weight codewords.
\end{lemma}

It follows from Lemmas \ref{weight} and \ref{weight2} that the weight distributions of $\mathcal{C}^{\perp}$ and $\cC$ can be   determined by the numbers of minimum weight codewords, and $A_{n-k}=A_{k}^{\bot}$. Let ${\cal C}_{n-k}$ be the set of minimum weight codewords in ${\cal C}$, and ${\cal H}_{\bf u}$ be the hyperplane in  PG$(k-1,q)$ determined by the nonzero vector ${\bf u}\in\F_q^k$. Then by Lemma \ref{Weight-lem1},  the following conclusions can be easily obtained:
\begin{equation}\label{weight-eq0}
	{\cal C}_{n-k}=\left\{{\bf u}G:{\bf 0}\neq{\bf u}\in\F_q^k,~\#({\cal H}_{\bf u}\cap S_G)=k\right\},
\end{equation}
and
\begin{equation}\label{weight-eq1}
	A_{n-k}=\#\left\{{\bf 0}\neq{\bf u}\in\F_q^k:\#({\cal H}_{\bf u}\cap S_G)=k\right\}.
\end{equation}
The above two results are easy but will play an important role in proving the main conclusions of the next sections.

\subsection{Ovals and Hyperovals}

Ovals and hyperovals are maximal arcs in PG$(2,q)$ for odd $q$ or even $q$, respectively.

When $q$ is odd, an oval in PG$(2,q)$ can be obtained in the following form:
\begin{equation}\label{oval-def}
	{\cal O}=\left\{(x^2,x,1):x\in\F_q \right\}\cup\{(1,0,0)\}.
\end{equation}

When $q$ is even, all hyperpvals in PG$(2,q)$ can be  constructed with a special type of permutation polynomials over $\F_q$, which is described in the following theorem \cite{Lidl1997}.

\begin{theorem}\label{ovalpoly}
Let $q \geq 2$ be a power of 2. Any hyperoval in  PG$(2, q)$ can be written in the following form
$$
\mathcal{HO}(f)=\{(f(c), c, 1): c \in \F_q\} \cup\{(1,0,0)\} \cup\{(0,1,0)\},
$$
where $f \in \F_q[x]$ is a permutation polynomial of $\F_q$ such that
 \begin{enumerate}
\item [(1)]  $deg(f)<q$ and $f(0)=0, f(1)=1$;
\item [(2)] for each $a \in \F_q, g_{a}(x):=(f(x+a)+f(a)) x^{q-2}$ is also a permutation polynomial of $\F_q$.
 \end{enumerate}
Conversely, every such set $\mathcal{HO}(f)$ is a hyperoval.
\end{theorem}

Polynomials satisfying the conditions of Theorem \ref{ovalpoly} are called  {\it o-polynomials}.
The following is a list of known infinite families of o-polynomials of $\F_q$ with $q=2^{m}$  in literature \cite{Maschietti1998}.
 \begin{enumerate}
\item [(1)] The translation polynomial $f(x)=x^{2^{h}}$, where $\operatorname{gcd}(h, m)=1$.
\item [(2)] The Segre polynomial $f(x)=x^{6}$, where $m$ is odd.
\item [(3)] The Glynn oval polynomial $f(x)=x^{3 \times 2^{(m+1) / 2}+4}$, where $m$ is odd.
\item [(4)] The Glynn oval polynomial $f(x)=x^{2^{(m+1) / 2}+2^{(m+1) / 4}}$ for $m \equiv 3~(\bmod ~4)$.
\item [(5)] The Glynn oval polynomial $f(x)=x^{2^{(m+1) / 2}+2^{(3 m+1) / 4}}$ for $m \equiv 1 ~(\bmod ~4)$.
\item [(6)] The Cherowitzo oval polynomial $f(x)=x^{2^{e}}+x^{2^{e}+2}+x^{3 \times 2^{e}+4}$, where $e=(m+1) / 2$ and $m$ is odd.
\item [(7)] The Payne oval polynomial $f(x)=x^{\frac{2^{m-1}+2}{3}}+x^{2^{m-1}}+x^{\frac{3 \times 2^{m-1}-2}{3}}$, where $m$ is odd.
\item [(8)] The Subiaco polynomial
$$
f_{a}(x)=\left(a^{2}\left(x^{4}+x\right)+a^{2}\left(1+a+a^{2}\right)\left(x^{3}+x^{2}\right)\right)\left(x^{4}+a^{2} x^{2}+1\right)^{2^{m}-2}+x^{2^{m-1}},
$$
where $\operatorname{Tr}_{q / 2}(1 / a)=1$ and $a \notin \F_4$ if $m \equiv 2 \bmod 4$.
\item [(9)] The Adelaide oval polynomial
$$
f(x)=\frac{T\left(\beta^{m}\right)(x+1)}{T(\beta)}+\frac{T\left(\left(\beta x+\beta^{q}\right)^{m}\right)}{T(\beta)\left(x+T(\beta) x^{2^{m-1}}+1\right)^{m-1}}+x^{2^{m-1}},
$$
where $m \geq 4$ is even, $\beta \in \F_{q^2} \backslash\{1\}$ with $\beta^{q+1}=1, m \equiv \pm(q-1) / 3~(\bmod ~q+1)$, and $T(x)=x+x^{q}$.
 \end{enumerate}

An useful  property on o-polynomials is also presented as follows.

\begin{lemma}[\cite{Maschietti1998}]\label{2-to-1}
 A polynomial $f$ over $\F_q$ with $f(0)=0$ is an oval polynomial if and only if $f_{u}:=f(x)+u x$ is 2-to-1 for every $u \in \F_q^{*}$.
\end{lemma}

\section{Two families of $[q+5,3,q+2]$ NMDS codes}\label{Sec: NMDS}

In this section, let $q=p^m$ with prime $p$ and positive integer $m$. We will present two constructions of $[q+5,3,q+2]$ NMDS codes according to $p=2$ or not, respectively. Their weight distributions are also determined.
This section will divided into two parts: one is based on the hyperovals in PG$(2,2^m)$, while another is based on an oval in PG$(2,p^m)$. First let
$\F_q=\{\alpha_1,\ldots,\alpha_{q-2},\alpha_{q-1},\alpha_{q}=0\}$.

\subsection{NMDS codes from hyperovals}

Let $q=2^m$.
Define $f\in\F_q[x]$ to be an o-polynomial  and  an element $v \in \F_q \setminus \{f(x)+x : x \in \F_q\}.$ By Lemma \ref{2-to-1}, there are $q/2$ choices of such $v$. Define a $3\times(q+5)$ matrix $G_v$ by
\begin{align}\label{generator}
G_v&=({\bf g}_1,\ldots,{\bf g}_{q+2},{\bf g}_{q+3},{\bf g}_{q+4},,{\bf g}_{q+5})	\notag\\
&=\left(\begin{array}{llllllllll}
f(\alpha_{1}) & f(\alpha_{2}) & \cdots & f(\alpha_{q-1})   & f(\alpha_{q}) & 1 & 0 & 1 & 0 & v\\
\alpha_{1} & \alpha_{2} & \cdots & \alpha_{q-1}            & \alpha_{q}    & 0 & 1 & 1 & v & 0\\
1 & 1 & \cdots & 1                                         & 1             & 0 & 0 & 0 & 1 & 1
\end{array}\right).
\end{align}
Obviously by Theorem \ref{ovalpoly}, the first $q+2$ columns of $G_v$ form a hyperoval in PG$(2,q)$. Let ${\cal C}_v$ be the linear code over $\F_q$ with the generator matrix $G_v$. We want to determine the parameters and weight enumerator of ${\cal C}_v$.

Firstly, we give the following results which will be used later.

\begin{lemma}\label{p even-lem}
Let $f\in\F_q[x]$ be an o-polynomial, $v \in \F_q \setminus \{f(x)+x : x \in \F_q\}$ and  $u_1,u_2\in\F_q^*$.  Define
\begin{align*}
	A_1(u_1,u_2)&=\#\{x\in\F_q:u_1f(x)+u_2x+u_2v=0\},\\
	A_2(u_1,u_2)&=\#\{x\in\F_q:u_1f(x)+u_2x+u_1v=0\}.
\end{align*}
Then $A_i(u_1,u_2)\in\{0,2\}$ for $i=1,2$. Especially, $A_i(u,u)=0$ for any $u\in \F_q^*$, and
\begin{equation*}
	\#\{(u_1,u_2)\in(\F_q^*)^2:A_i(u_1,u_2)=2\}=(q-1)(q-2)/2.
\end{equation*}
\end{lemma}
\proof We first prove the results for the first equation
\begin{equation}\label{p even-lem-eq1}
	u_1f(x)+u_2x+u_2v=0,~{\rm for}~ u_1,u_2\in\F_q^*.
\end{equation}

By Lemma \ref{2-to-1}, the function $u_1f(x)+u_2x$ is 2-to-1 for any given $u_1,u_2\in\F_q^*$. Thus $A_1(u_1,u_2)\in\{0,2\}$.
This means that if (\ref{p even-lem-eq1}) has solutions, they appear in pairs. The result $A_i(u,u)=0$ is obvious since $v \in \F_q \setminus \{f(x)+x : x \in \F_q\}$.

Note that the choice of $v$ implies that 0 and $v$ can not be solutions of
  (\ref{p even-lem-eq1}) since $v\neq0$ and $f(v)\neq0$. Thus when $x$ runs through $\F_q\setminus\{0,v\}$, $\frac{u_1}{u_2}=\frac{v+x}{f(x)}$ takes exactly $(q-2)/2$ different values.  Conversely, if $\frac{u_1}{u_2}$ takes one of such  $(q-2)/2$  values,  (\ref{p even-lem-eq1}) has exactly two solutions. Therefore, there are exactly $(q-1)(q-2)/2$ choices of $(u_1,u_2)$ such that (\ref{p even-lem-eq1}) has exactly two solutions.

The proof for the results of the second equation is almost the same, so the details are omitted. We only need to notice that 0 and $f^{-1}(v)$ can not be solutions of the second equation.
\qed

\begin{theorem}\label{q+5}
Let $q=2^m$. The linear code ${\cal C}_v$ is an $[q+5,3,q+2]$ NMDS code over $\F_q$ with weight enumerator
$$
A(z)=1 + \frac{(q-1)(3q+8)}{2} z^{q+2} + \frac{(q-1)(q+2)(q-2)}{2} z^{q+3} +  \frac{3(q-1)(q-2)}{2} z^{q+4} + \frac{(q-1)(q-2)^{2}}{2} z^{q+5}.
$$

%$$公式太长可以分段：
%\begin{aligned}
%A(z)=1 + \frac{(q-1)(3q+8)}{2} z^{q+2} + \frac{(q-1)(q+2)(q-2)}{2} z^{q+3} &+  \frac{3(q-1)(q-2)}{2} z^{q+4}\\
%&+ \frac{(q-1)(q-2)^{2}}{2} z^{q+5}.
%\end{aligned}
%$$
\end{theorem}
\proof
Let $S_v$ be the set of columns of $G_v$. In order to prove that ${\cal C}_v$ is an NMDS codes, we only need to prove that $S_v$ is an $(n,3)$-arc in PG$(2,q)$.

It's easy to find that any two points of $S_v$ generate a line in PG$(2,q)$, and there exist three collinear points, such as $\{{\bf g}_{q+1},{\bf g}_{q+2},{\bf g}_{q+3}\}$. Now we will prove that no four points of $S_v$ lie on a line, i.e., $\#({\cal H}_{\bf u}\cap S_v)\leq 3$ for any line ${\cal H}_{\bf u}$ with nonzero ${\bf u}=(u_1,u_2,u_3)\in \F_q^3$.

Since $S_v'=\{{\bf g}_1,\ldots,{\bf g}_{q+2}\}$ is a hyperoval, which means  $\#({\cal H}_{\bf u}\cap S_v')\leq 2$ for any line ${\cal H}_{\bf u}$. Thus we only consider the intersection number $\#({\cal H}_{\bf u}\cap S_v)$ for any line ${\cal H}_{\bf u}$ which is incident with some of $\{{\bf g}_{q+3},{\bf g}_{q+4},{\bf g}_{q+5}\}$. The rest of the proof will be divided into three cases:

{\bf Case 1}. ${\bf g}_{q+3}\in{\cal H}_{\bf u}$, i.e., $u_1+u_2=0$, which implies that ${\bf u}=(u_1,u_1,u_3)$.
\begin{itemize}
\item	If $u_1=0$, then ${\bf u}=(0,0,u_3)$ for $u_3\in\F_q^*$. Obviously ${\cal H}_{\bf u}\cap S_v=\{{\bf g}_{q+1},{\bf g}_{q+2},{\bf g}_{q+3}\}$.
\item If $u_1\neq0$ and $u_3=u_1v$, then ${\bf u}=(u_1,u_1,u_1v)$. It's easy to obtain that $\{{\bf g}_{q+3},{\bf g}_{q+4},{\bf g}_{q+5}\}\in{\cal H}_{\bf u}$. The choice of $v$ makes that $f(x)+x+v=0$ has no solution in $\F_q$, thus ${\cal H}_{\bf u}\cap S_v=\{{\bf g}_{q+3},{\bf g}_{q+4},{\bf g}_{q+5}\}$.

%\item If $u_1\neq0$ and $u_3\in \F_q\setminus \{u_1v\}$ satisfying that $f(x)+x+u_1^{-1}u_3=0$ have solutions in $\F_q$. Since that $f(x)+x$ is 2-to-1 by Lemma \ref{2-to-1}, it has exactly two solutions $\{\alpha_i,\alpha_j\}$. Thus ${\cal H}_{\bf u}\cap S_v=\{{\bf g}_{i},{\bf g}_{j},{\bf g}_{q+3}\}$. There are $q/2$ choices of such $u_3$ for given $u_1$ since $f(x)+x$ is 2-to-1.
%\item If $u_1\neq0$ and $u_3\in \F_q\setminus \{u_1v\}$ satisfying that $f(x)+x+u_1^{-1}u_3=0$ have no solutions in $\F_q$. Now ${\cal H}_{\bf u}\cap S_v=\{{\bf g}_{q+3}\}$.

	\item If $u_1\neq0$ and $u_3\in \F_q\setminus \{u_1v\}$, then ${\bf g}_k\notin{\cal H}_{\bf u}$ for $k\in\{q+1,q+2,q+4,q+5\}$. Now we consider the equation $f(x)+x+u_1^{-1}u_3=0$.

	\begin{itemize}
	\item If $f(x)+x+u_1^{-1}u_3=0$ has solutions in $\F_q$. Since that $f(x)+x$ is 2-to-1 by Lemma \ref{2-to-1}, it has exactly two solutions $\{\alpha_{i_1},\alpha_{j_1}\}$. Thus ${\cal H}_{\bf u}\cap S_v=\{{\bf g}_{i_1},{\bf g}_{j_1},{\bf g}_{q+3}\}$. There are $q(q-1)/2$ choices of $(u_1,u_3)$ satisfying that  $u_1^{-1}u_3\in \{f(x)+x:x\in\F_q\}$.

	\item If $f(x)+x+u_1^{-1}u_3=0$ has no solution in $\F_q$, then ${\cal H}_{\bf u}\cap S_v=\{{\bf g}_{q+3}\}$.
	\end{itemize}

\end{itemize}

 Thus in this case, the number of ${\bf u}\in\F_q^3$ satisfying $\#({\cal H}_{\bf u}\cap S_v)=3$ is exactly
\[(1+1+q/2)(q-1)=(q+4)(q-1)/2.\]

{\bf Case 2}. ${\bf g}_{q+3}\notin{\cal H}_{\bf u}$ and ${\bf g}_{q+4}\in{\cal H}_{\bf u}$. Then ${\bf u}=(u_1,u_2,u_2v)$ with $u_1\neq u_2$.

\begin{itemize}
	\item If $u_2=0$, then ${\bf u}=(u_1,0,0)$ for $u_1\in\F_q^*$. Obviously ${\cal H}_{\bf u}\cap S_v=\{{\bf g}_{q},{\bf g}_{q+2},{\bf g}_{q+4}\}$.
	\item If $u_2\neq0, u_1=0$, then ${\bf u}=(0,u_2,u_2v)$ for $u_2\in\F_q^*$.  It's easy to obtain that ${\cal H}_{\bf u}\cap S_v=\{{\bf g}_{i},{\bf g}_{q+1},{\bf g}_{q+4}\}$, where $i$ is the unique number such that $\alpha_i=v$.
	\item If $u_2\neq0$, $u_1\neq0$, then ${\bf g}_k\notin{\cal H}_{\bf u}$ for $k\in\{q+1,q+2,q+3,q+5\}$ because of $u_1\neq u_2$. Consider the equation $u_1f(x)+u_2x+u_2v=0$.
	\begin{itemize}
	%\item If $u_1f(x)+u_2x+u_2v=0$ has exactly two solutions $\{x_1,x_2\}$ in $\F_q$, then ${\cal H}_{\bf u}\cap S_v=\{{\bf g}_{i},{\bf g}_{j},{\bf g}_{q+4}\}$, where $i,j$ are the unique numbers such that $\alpha_i=x_1,\alpha_j=x_2$. By Lemma \ref{p even-lem}, there are $(q-1)(q-2)/2$ choices of such $(u_1,u_2)$.
	\item If $u_1f(x)+u_2x+u_2v=0$ has exactly two solutions $\{\alpha_{i_2},\alpha_{j_2}\}$ in $\F_q$, then ${\cal H}_{\bf u}\cap S_v=\{{\bf g}_{i_2},{\bf g}_{j_2},{\bf g}_{q+4}\}$. There are $(q-1)(q-2)/2$ choices of such $(u_1,u_2)$ by Lemma \ref{p even-lem}.
	\item If the function $u_1f(x)+u_2x+u_2v=0$ has no solution in $\F_q$, then ${\cal H}_{\bf u}\cap S_v=\{{\bf g}_{q+4}\}$.
	\end{itemize}

\end{itemize}

%If the function $u_1f(x)+u_2x=u_2v$ has exactly two solutions $\{x_1,x_2\}$ in $\F_q$, then  ${\cal H}_u\cap S_v=\{{\bf g}_{i},{\bf g}_{j},{\bf g}_{q+4}\}$, where $i,j$ are the unique numbers such that $\alpha_i=x_1,\alpha_j=x_2$. The choice of $v$ makes sure that $x_1,x_2\in\F_q\setminus \{0,v\}$. Since $f(x)+\frac{u_2}{u_1}x$ is 2-to-1, when $x$ runs through $\F_q\setminus\{0,v\}$, $\frac{u_2}{u_1}$ takes exactly $(q-2)/2$ different values.
%
%If the function $u_1f(x)+u_2x=u_2v$ has no solutions in $\F_q$, then ${\cal H}_u\cap S_v=\{{\bf g}_{q+4}\}$.

In this case, the number of ${\bf u}\in\F_q^3$ satisfying $\#({\cal H}_{\bf u}\cap S_v)=3$ is exactly
\[(1+1+(q-2)/2)(q-1)=(q+2)(q-1)/2.\]

{\bf Case 3}. ${\cal H}_{\bf u}\cap \{{\bf g}_{q+3},{\bf g}_{q+4},{\bf g}_{q+5}\}={\bf g}_{q+5}$. Then ${\bf u}=(u_1,u_2,u_1v)$ with $u_1\neq u_2$.
\begin{itemize}
	\item If $u_1=0$, then ${\bf u}=(0,u_2,0)$ for $u_2\in\F_q^*$. Obviously ${\cal H}_{\bf u}\cap S_v=\{{\bf g}_{q},{\bf g}_{q+1},{\bf g}_{q+5}\}$.
	\item If $u_1\neq0, u_2=0$, then ${\bf u}=(u_1,0,u_1v)$ for $u_1\in\F_q^*$.  It's easy to obtain that ${\cal H}_{\bf u}\cap S_v=\{{\bf g}_{j},{\bf g}_{q+2},{\bf g}_{q+5}\}$, where $j$ is the unique number such that $f(\alpha_j)=v$.
	\item If $u_1\neq0$, $u_2\neq0$, then ${\bf g}_k\notin{\cal H}_{\bf u}$ for $k\in\{q+1,q+2,q+3,q+4\}$ because of $u_1\neq u_2$. Consider the equation $u_1f(x)+u_2x+u_1v=0$.
    \begin{itemize}
   % \item If $u_1f(x)+u_2x+u_1v=0$ has exactly two solutions $\{x_1',x_2'\}$ in $\F_q$, then  ${\cal H}_{\bf u}\cap S_v=\{{\bf g}_{i'},{\bf g}_{j'},{\bf g}_{q+4}\}$, where $i',j'$ are the unique numbers such that $\alpha_{i'}=x_1',\alpha_{j'}=x_2'$. By Lemma \ref{p even-lem}, there are $(q-1)(q-2)/2$ choices of such $(u_1,u_2)$.
    \item If $u_1f(x)+u_2x+u_1v=0$ has exactly two solutions $\{\alpha_{i_3},\alpha_{j_3}\}$ in $\F_q$, then  ${\cal H}_{\bf u}\cap S_v=\{{\bf g}_{i_3},{\bf g}_{j_3},{\bf g}_{q+4}\}$. There are $(q-1)(q-2)/2$ choices of such $(u_1,u_2)$ by Lemma \ref{p even-lem}.
    \item If $u_1f(x)+u_2x+u_1v=0$ has no solution in $\F_q$, then ${\cal H}_{\bf u}\cap S_v=\{{\bf g}_{q+5}\}$.
   \end{itemize}
\end{itemize}

In this case, the number of ${\bf u}\in\F_q^3$ satisfying $\#({\cal H}_{\bf u}\cap S_v)=3$ is exactly
\[(1+1+(q-2)/2)(q-1)=(q+2)(q-1)/2.\]

Based on the above cases, we prove that $\#({\cal H}_{\bf u}\cap S_v)\leq 3$ for any nonzero ${\bf u}\in\F_q^3$, therefore we obtain that $S_v$ is an $(q+5,3)$-arc in PG$(2,q)$, and ${\cal C}_v$ is an $[q+5,3,q+2]$ NMDS code over $\F_q$. In order to determine the weight distribution of the resultant NMDS code, we only need to compute $A_{q+2}$, which is equal to
\begin{align*}
	A_{q+2}&=\#\left\{{\bf u}\in \F_q^3\setminus\{{\bf0}\}:\#({\cal H}_{\bf u}\cap S_v)=3\right\}\\
	&=(q+4+q+2+q+2)(q-1)/2=(3q+8)(q-1)/2.
\end{align*}
By Lemma \ref{weight}, the weight enumerator of ${\cal C}_v$  can be completely determined.
\qed

Now we give an example to illustrate Theorem \ref{q+5}.

\begin{example}\label{ex1}

Let $q=2^2$, $\xi$ be a generator of $\F_{2^2}^*$ satisfying $\xi^2+\xi+1=0$. Choose the o-polynomial as $f(x)=x^2$ in Theorem \ref{q+5}, and $v=\xi$, then
%
%Let $q=2^4$, $\F_{2^4}=\F_2[x]/\langle x^4+x+1\rangle$, the oval polynomial $f(x)=x^2$. Let $\xi$ denote the generator of the multiplicative cyclic group $\F_{2^4}^*$. Then $ \F_{2^4} \setminus \{x^2+x : x \in \F_{2^4}\}= \{ \xi^3,\xi^6,\xi^7,\xi^9 ,\xi^{11}, \xi^{12}, \xi^{13},\xi^{14} \}.$ Without loss of generality, take $v=\xi^{9}$,
\begin{equation}\label{ex1-1}
G_v=\left(\begin{array}{lllllllll}
\xi& \xi^{2}& 1& 0 &1 &0& 1& 0 &\xi\\
\xi^{2}& \xi& 1 &0 &0 &1& 1& \xi& 0\\
1 &1 &1 &1& 0& 0& 0& 1 &1
\end{array}\right).
\end{equation}
By Magma \cite{Magma}, it is easy to check that
linear code $\cC_v$ generated by (\ref{ex1-1}) is a $[9,3,6]$  NMDS code with weight enumerator
$$
A(z)=1 + 30 z^{6} + 18 z^{7} +  9 z^{8} + 6 z^{9},
$$
which coincides with the conclusion of Theorem \ref{q+5}.
\end{example}

\subsection{NMDS codes from ovals}

Let $q=p^m$ with an odd prime $p$, and $\eta(x)$ be the quadratic character function of $\F_q$. The following are some basic knowledges of $\eta(x)$ \cite{Lidl1997}.

\begin{lemma}\label{quadratic-lem}
\begin{enumerate}
	\item $\eta(xy)=\eta(x)\eta(y)$ for any $x,y \in \F_q$;
	\item $\sum_{x \in \F_q} \eta(x)=0$;
	\item $\eta(-1)=
\begin{cases}
1,  & \text { if } q \equiv 1 ~(mod~4), \\
-1, & \text { if }q \equiv 3 ~(mod~4);
\end{cases}$
\item $\sum_{x \in \F_q} \eta(ax^2+bx+c)=-\eta(a)$ for any $a \in \F_q^{*}$, $b,c \in \F_q$.
\end{enumerate}	
\end{lemma}

Let $w\in \F_q$ satisfying $\eta(w)=\eta(1+4w)=-1$. Define a $3\times (q+5)$ matrix $G_w$ by
\begin{align}\label{generator2}
G_w&=({\bf g}_1,\ldots,{\bf g}_{q+1},{\bf g}_{q+2},{\bf g}_{q+3},{\bf g}_{q+4},{\bf g}_{q+5})	\notag\\
&=\left(\begin{array}{llllllllll}
\alpha_{1}^2 & \alpha_{2}^2 & \cdots & \alpha_{q-1}^2   & \alpha_{q}^2     & 1 & 0 & 1 & 0  & w\\
\alpha_{1} & \alpha_{2} & \cdots & \alpha_{q-1}            & \alpha_{q}    & 0 & 1 & 1 & w  & 0\\
1 & 1 & \cdots & 1                                         & 1             & 0 & 0 & 0 & -1 & 1
\end{array}\right).
\end{align}
Obviously, the first $q+1$ columns of $G_w$ form a oval in PG$(2,q)$. Let ${\cal C}_w$ be the linear code over $\F_q$ with the generator matrix $G_w$. We want to determine the parameters and weight enumerators of ${\cal C}_w$.

We first demonstrate the existence of such $w$ in $\F_q$.

\begin{lemma}\label{num2}
There are $\frac{q-2+\eta(-1)}{4}$ choices of such $w$ in $\F_q$.
\end{lemma}
\begin{proof}
Let $N_w=\#\{w \in \F_q : \eta(w)=\eta(1+4w)=-1\}$. Then by   Lemma \ref{quadratic-lem}, we have
\[\begin{aligned}
N_w &= \frac{1}{4} \sum\limits_{w\neq 0, -\frac{1}{4}} (1-\eta(w))(1-\eta(1+4w))\\
  &= \frac{1}{4} \sum\limits_{w \in \F_q} (1-\eta(w))(1-\eta(1+4w)) -\frac{1}{4}(1-\eta(0))(1-\eta(1)) -\frac{1}{4}(1-\eta(-\frac{1}{4}))(1-\eta(0))\\
  &= \frac{1}{4} \sum\limits_{w \in \F_q} (1-\eta(w)-\eta(1+4w)+\eta(w(1+4w))) -\frac{1}{4}(1-\eta(-1)) \\
  &= \frac{1}{4} \sum\limits_{w \in \F_q} 1  + \frac{1}{4} \sum\limits_{w \in \F_q} \eta(4w^2+w) -\frac{1}{4}(1-\eta(-1))\\
  &= \frac{q-2+\eta(-1)}{4}.
\end{aligned}
\]
\end{proof}

Then we give the following results which will be used later.

\begin{lemma}\label{p odd-lem}
Let $w \in \F_q$ satisfying $\eta(w)=\eta(1+4w)=-1$, and $u_1,u_2\in\F_q^*$.  Define
\begin{align*}
	B_1(u_1,u_2)&=\#\{x\in\F_q:u_1x^2+u_2x+u_2w=0\},\\
	B_2(u_1,u_2)&=\#\{x\in\F_q:u_1x^2+u_2x-u_1w=0\}.
\end{align*}
Then $B_i(u_1,u_2)\in\{0,1,2\}$ for $i=1,2$. Especially, $A_i(u,-u)=0$ for any $u\in \F_q^*$, and
\begin{equation*}
	\#\{(u_1,u_2)\in(\F_q^*)^2:B_1(u_1,u_2)=2\}=(q-1)(q-3)/2.
\end{equation*}	
\begin{equation*}
	\#\{(u_1,u_2)\in(\F_q^*)^2:B_2(u_1,u_2)=2\}=(q-1)(q-2+\eta(-1))/2.
\end{equation*}	
\end{lemma}

\proof
It's obvious that $B_i(u_1,u_2)\in\{0,1,2\}$ for $i=1,2$. The result $B_i(u,-u)=0$ is since $\eta(1+4w)=-1$.

For the first equation
\begin{equation}\label{p odd-lem-eq1}
	u_1x^2+u_2x+u_2w=0,~{\rm for}~ u_1,u_2\in\F_q^*.
\end{equation}

Let $\mu=\frac{u_2}{u_1}$, $N_{\mu}= \# \{\mu \in \F_q: \eta(\mu^2-4\mu w)=1\}$. By Lemma \ref{quadratic-lem},
$$
\begin{aligned}
N_{\mu} &= \frac{1}{2} \sum\limits_{\mu \neq 0, 4w} (1+\eta(\mu ^2-4\mu w))\\
  &= \frac{1}{2} \sum\limits_{\mu \in \F_q} (1+\eta(\mu ^2-4\mu w)) -\frac{1}{2}(1+\eta(0))-\frac{1}{2}(1+\eta(0))\\
  &= \frac{1}{2} \sum\limits_{\mu \in \F_q} 1  + \frac{1}{2} \sum\limits_{\mu \in \F_q} \eta(\mu ^2-4\mu w) -1\\
  &= \frac{q}{2} - \frac{1}{2} - 1\\
  &= \frac{q-3}{2}.
\end{aligned}
$$
Therefore, there are exactly $(q-1)(q-3)/2$ choices of $(u_1,u_2)$ such that (\ref{p odd-lem-eq1}) has exactly two solutions.

For the second equation
\begin{equation}\label{p odd-lem-eq2}
	u_1x^2+u_2x-u_1w=0,~{\rm for}~ u_1,u_2\in\F_q^*.
\end{equation}
Let $\nu=\frac{u_2}{u_1}$, $N_{\nu}= \# \{\nu \in \F_q: \eta(\nu^2+4w)=1\}$. By Lemma \ref{quadratic-lem},
$$
\begin{aligned}
N_{\nu} &= \frac{1}{2} \sum\limits_{\nu^2\neq -4w} (1+\eta(\nu^2+4w))\\
  &= \frac{1}{2} \sum\limits_{\nu \in \F_q} (1+\eta(\nu^2+4w)) -\frac{1}{2} \sum\limits_{\nu^2= -4w} 1\\
  &= \frac{1}{2} \sum\limits_{\nu \in \F_q} 1  + \frac{1}{2} \sum\limits_{\nu \in \F_q} \eta(\nu^2+4w) -\frac{1-\eta(-1)}{2}\\
  &= \frac{q}{2} - \frac{1}{2} - \frac{1-\eta(-1)}{2}\\
  &= \frac{q-2+\eta(-1)}{2}.
\end{aligned}
$$
Therefore, there are exactly $(q-1)(q-2+\eta(-1))/2$ choices of $(u_1,u_2)$ such that (\ref{p odd-lem-eq2}) has exactly two solutions.
\qed

\begin{theorem}\label{q+5odd}
 Let $q=p^m$, where $p$ is an odd prime and $m$ is a positive integer. The linear code $\cC_w$ is an $[q+5, 3, q+2]$ NMDS code over $\F_q$.

 If $q \equiv 1 ~(mod ~4)$, the weight enumerator of $\cC_w$  is
$$
A(z)=1 + (2q+2)(q-1) z^{q+2} + \frac{(q-1)(q^2-3q+8)}{2} z^{q+3} + (3q-9)(q-1) z^{q+4} + \frac{(q-1)(q^2-5q+8)}{2} z^{q+5}.
$$

%$$公式太长可以分段：
%\begin{aligned}
%A(z)=1 + (2q+2)(q-1) z^{q+2} + \frac{(q-1)(q^2-3q+8)}{2} z^{q+3} &+ (3q-9)(q-1) z^{q+4}\\
%&+ \frac{(q-1)(q^2-5q+8)}{2} z^{q+5}.
%\end{aligned}
%$$

If $q \equiv 3 ~(mod ~4)$, the weight enumerator of $\cC_w$ is
$$
A(z)=1 + (2q+1)(q-1) z^{q+2} + \frac{(q-1)(q^2-3q+14)}{2} z^{q+3} +  (3q-12)(q-1) z^{q+4} + \frac{(q-1)(q^2-5q+10)}{2} z^{q+5}.
$$

%$$公式太长可以分段：
%\begin{aligned}
%A(z)=1 + (2q+1)(q-1) z^{q+2} + \frac{(q-1)(q^2-3q+14)}{2} z^{q+3} &+  (3q-12)(q-1) z^{q+4}\\
%&+ \frac{(q-1)(q^2-5q+10)}{2} z^{q+5}.
%\end{aligned}
%$$
\end{theorem}
\proof
Let $S_w$ be the set of columns of $G_w$. In order to prove that ${\cal C}_w$ is an NMDS code, we only need to prove that $S_w$ is an $(n,3)$-arc in PG$(2,q)$.

It's easy to find that any two points of $S_w$ generate a line in PG$(2,q)$, and there exist three collinear points, such as $\{{\bf g}_{q+1},{\bf g}_{q+2},{\bf g}_{q+3}\}$. Now we will prove that no four points of $S_w$ lie on a line, i.e., $\#({\cal H}_{\bf u}\cap S_w)\leq 3$ for any line ${\cal H}_{\bf u}$ with nonzero ${\bf u}=(u_1,u_2,u_3)\in \F_q^3$.

Since $S_w'=\{{\bf g}_1,\ldots,{\bf g}_{q+1}\}$ is an oval, which means $\#({\cal H}_{\bf u}\cap S_w')\leq 2$ for any line ${\cal H}_{\bf u}$. Thus we only consider the intersection number $\#({\cal H}_{\bf u}\cap S_w)$ for any line ${\cal H}_{\bf u}$ which is incident with some of $\{{\bf g}_{q+2},{\bf g}_{q+3},{\bf g}_{q+4},{\bf g}_{q+5}\}$. The rest of the proof will be divided into four cases:

{\bf Case 1}. ${\bf g}_{q+2}\in{\cal H}_{\bf u}$, i.e., $u_2=0$, which implies that ${\bf u}=(u_1,0,u_3)$.
\begin{itemize}
\item	If $u_1=0$, then ${\bf u}=(0,0,u_3)$ for $u_3\in\F_q^*$. Obviously ${\cal H}_{\bf u}\cap S_w=\{{\bf g}_{q+1},{\bf g}_{q+2},{\bf g}_{q+3}\}$.
\item If $u_1\neq 0$ and $u_3=0$, then ${\bf u}=(u_1,0,0)$ for $u_1\in\F_q^*$. Obviously ${\cal H}_{\bf u}\cap S_w=\{{\bf g}_{q},{\bf g}_{q+2},{\bf g}_{q+4}\}$.
\item If $u_1\neq0$ and $u_3=-u_1w$, then ${\bf u}=(u_1,0,-u_1w)$. It's easy to obtain that $\{{\bf g}_{q+2},{\bf g}_{q+5}\}\in{\cal H}_{\bf u}$. The choice of $w$ makes that $x^2-w=0$ has no solution in $\F_q$, thus ${\cal H}_{\bf u}\cap S_w=\{{\bf g}_{q+2},{\bf g}_{q+5}\}$.

	\item If $u_1\neq0$ and $u_3\in \F_q\setminus \{0,-u_1w\}$, then ${\bf g}_k\notin{\cal H}_{\bf u}$ for $k\in\{q+1,q+3,q+4,q+5\}$. Now we consider the equation $x^2+u_1^{-1}u_3=0$.

	\begin{itemize}
	\item If $x^2+u_1^{-1}u_3=0$ has solutions in $\F_q$. Since $u_1^{-1}u_3\neq 0$, it has exactly two solutions $\{\alpha_{i_1},\alpha_{j_1}\}$. Thus ${\cal H}_{\bf u}\cap S_w=\{{\bf g}_{i_1},{\bf g}_{j_2},{\bf g}_{q+2}\}$. There are $(q-1)/2$ choices of such $u_3$ for given $u_1$ since $\eta(-y)=1$ has $(q-1)/2$ solutions.

	\item If $x^2+u_1^{-1}u_3=0$ has no solution in $\F_q$, then ${\cal H}_{\bf u}\cap S_w=\{{\bf g}_{q+2}\}$.
	\end{itemize}

\end{itemize}

 Thus in this case, the number of ${\bf u}\in\F_q^3$ satisfying $\#({\cal H}_{\bf u}\cap S_w)=3$ is exactly
\[(1+1+(q-1)/2)(q-1)=(q+3)(q-1)/2.\]

{\bf Case 2}. ${\bf g}_{q+2}\notin{\cal H}_{\bf u}$ and ${\bf g}_{q+3}\in{\cal H}_{\bf u}$. Then ${\bf u}=(u_1,-u_1,u_3)$ with $u_1\neq 0$.
\begin{itemize}
\item If $u_3=-u_1w$, then ${\bf u}=(u_1,-u_1,-u_1w)$ for $u_1\in\F_q^*$. Obviously ${\cal H}_{\bf u}\cap S_w=\{{\bf g}_{q+3},{\bf g}_{q+4},{\bf g}_{q+5}\}$.

\item If $u_3\in \F_q\setminus \{-u_1w\}$ then ${\bf g}_k\notin{\cal H}_{\bf u}$ for $k\in\{q+1,q+2,q+4,q+5\}$. Consider the equation $x^2-x+u_1^{-1}u_3=0$.

	\begin{itemize}

	\item If $x^2-x+u_1^{-1}u_3=0$ has exactly two solutions $\{\alpha_{i_2},\alpha_{j_2}\}$ in $\F_q$, then ${\cal H}_u\cap S_w=\{{\bf g}_{i_2},{\bf g}_{j_2},{\bf g}_{q+3}\}$. There are $(q-1)/2$ choices of such $u_3$ for given $u_1$ since $\eta(1-4y)=1$ has $(q-1)/2$ solutions.
	\item If $x^2-x+u_1^{-1}u_3=0$ has only one solution $\alpha_{i}$ in $\F_q$, then ${\cal H}_u\cap S_w=\{{\bf g}_{i},{\bf g}_{q+3}\}$. %where $i$ is the unique number such that $\alpha_{i}=x$.
    \item If $x^2-x+u_1^{-1}u_3=0$ has no solution in $\F_q$, then ${\cal H}_{\bf u}\cap S_w=\{{\bf g}_{q+3}\}$.
    \end{itemize}

\end{itemize}

 Thus in this case, the number of ${\bf u}\in\F_q^3$ satisfying $\#({\cal H}_{\bf u}\cap S_w)=3$ is exactly
\[(1+(q-1)/2)(q-1)=(q+1)(q-1)/2.\]

{\bf Case 3}. ${\bf g}_{q+2},~{\bf g}_{q+3}\notin{\cal H}_{\bf u}$ and ${\bf g}_{q+4}\in{\cal H}_{\bf u}$. Then ${\bf u}=(u_1,u_2,u_2w)$ with $-u_1\neq u_2\neq 0$.

\begin{itemize}
	\item If $u_1=0$, then ${\bf u}=(0,u_2,u_2w)$ for $u_2\in\F_q^*$.  It's easy to obtain that ${\cal H}_{\bf u}\cap S_w=\{{\bf g}_{j},{\bf g}_{q+1},{\bf g}_{q+4}\}$, where $j$ is the unique number such that $\alpha_j=-w$.
	\item If $u_1\neq0$, then ${\bf g}_k\notin{\cal H}_{\bf u}$ for $k\in\{q+1,q+2,q+3,q+5\}$. Consider the equation $u_1x^2+u_2x+u_2w=0$.
	\begin{itemize}
	%\item If $u_1x^2+u_2x+u_2w=0$ has exactly two solutions $\{x_1,x_2\}$ in $\F_q$, then  ${\cal H}_{\bf u}\cap S_w=\{{\bf g}_{i},{\bf g}_{j},{\bf g}_{q+4}\}$, where $i,j$ are the unique numbers such that $\alpha_i=x_1,\alpha_j=x_2$. By Lemma \ref{p odd-lem}, there are $(q-1)(q-3)/2$ choices of such $(u_1,u_2)$.
	\item If $u_1x^2+u_2x+u_2w=0$ has exactly two solutions $\{\alpha_{i_3},\alpha_{j_3}\}$ in $\F_q$, then  ${\cal H}_{\bf u}\cap S_w=\{{\bf g}_{i_3},{\bf g}_{j_3},{\bf g}_{q+4}\}$. There are $(q-1)(q-3)/2$ choices of such $(u_1,u_2)$ by Lemma \ref{p odd-lem}.
	\item If $u_1x^2+u_2x+u_2w=0$ has only one solution $\alpha_l$ in $\F_q$, then ${\cal H}_{\bf u}\cap S_w=\{{\bf g}_{l},{\bf g}_{q+4}\}$. %where $l$ is the unique number such that $\alpha_l=x$.
	\item If $u_1x^2+u_2x+u_2w=0$ has no solution in $\F_q$, then ${\cal H}_{\bf u}\cap S_w=\{{\bf g}_{q+4}\}$.
	\end{itemize}

\end{itemize}

In this case, the number of ${\bf u}\in\F_q^3$ satisfying $\#({\cal H}_{\bf u}\cap S_w)=3$ is exactly
\[(1+(q-3)/2)(q-1)=(q-1)^2/2.\]

{\bf Case 4}. ${\cal H}_{\bf u}\cap \{{\bf g}_{q+2},{\bf g}_{q+3},{\bf g}_{q+4},{\bf g}_{q+5}\}={\bf g}_{q+5}$. Then ${\bf u}=(u_1,u_2,-u_1w)$ with $-u_1\neq u_2 \neq 0$.
\begin{itemize}
	\item If $u_1=0$, then ${\bf u}=(0,u_2,0)$ for $u_2\in\F_q^*$. Obviously ${\cal H}_{\bf u}\cap S_w=\{{\bf g}_{q},{\bf g}_{q+1},{\bf g}_{q+5}\}$.
	\item If $u_1\neq0$, then ${\bf g}_k\notin{\cal H}_u$ for $k\in\{q+1,q+2,q+3,q+4\}$. Consider the equation $u_1x^2+u_2x-u_1w=0$.

    \begin{itemize}
	%\item If $u_1x^2+u_2x-u_1w=0$ has exactly two solutions $\{x_1,x_2\}$ in $\F_q$, then ${\cal H}_{\bf u}\cap S_w=\{{\bf g}_{i},{\bf g}_{j},{\bf g}_{q+5}\}$, where $i,j$ are the unique numbers such that $\alpha_i=x_1,\alpha_j=x_2$. By Lemma \ref{p odd-lem}, there are $(q-1)(q-2+\eta(-1))/2$ choices of such $(u_1,u_2)$.
	\item If $u_1x^2+u_2x-u_1w=0$ has exactly two solutions $\{\alpha_{i_4},\alpha_{j_4}\}$ in $\F_q$, then ${\cal H}_{\bf u}\cap S_w=\{{\bf g}_{i_4},{\bf g}_{j_4},{\bf g}_{q+5}\}$. There are $(q-1)(q-2+\eta(-1))/2$ choices of such $(u_1,u_2)$ by Lemma \ref{p odd-lem}.
	\item If $u_1x^2+u_2x-u_1w=0$ has only one solution $\alpha_t$ in $\F_q$, then ${\cal H}_{\bf u}\cap S_w=\{{\bf g}_{t},{\bf g}_{q+5}\}$. %where $i$ is the unique number such that $\alpha_i=x$.
	\item If $u_1x^2+u_2x-u_1w=0$ has no solution in $\F_q$, then ${\cal H}_{\bf u}\cap S_w=\{{\bf g}_{q+5}\}$.
	\end{itemize}
\end{itemize}

In this case, the number of ${\bf u}\in\F_q^3$ satisfying $\#({\cal H}_{\bf u}\cap S_w)=3$ is exactly
\[(1+(q-2+\eta(-1))/2)(q-1)=(q+\eta(-1))(q-1)/2.\]

Based on the above cases, we  prove that $\#({\cal H}_{\bf u}\cap S_w)\leq 3$ for any nonzero ${\bf u}\in\F_q^3$, thus we obtain that $S_w$ is an $(q+5,3)$-arc in PG$(2,q)$, and ${\cal C}_w$ is an $[q+5,3,q+2]$ NMDS code over $\F_q$. In order to determine the weight distribution of the resultant NMDS code, we only need to compute $A_{q+2}$, which is equal to
\begin{align*}
	A_{q+2}&=\#\left\{{\bf u}\in \F_q^3\setminus\{{\bf0}\}:\#({\cal H}_{\bf u}\cap S_w)=3\right\}\\
	&=(q+3+q+1+q-1+q+\eta(-1))(q-1)/2\\
	&=(4q+3+\eta(-1))(q-1)/2\\
	&=\begin{cases}
(2q+2)(q-1),  & \text { if } q \equiv 1 ~(mod~4), \\
(2q+1)(q-1), & \text { if }q \equiv 3 ~(mod~4).
\end{cases}
\end{align*}
The last equation is obtained by
Lemma \ref{quadratic-lem} (3). Therefore,
the weight distribution of ${\cal C}_w$ can be completely determined by Lemma \ref{weight}.
\qed

Now we give two examples to illustrate Theorem \ref{q+5odd}.

\begin{example}\label{ex2}
Let $q=3^2$, $\zeta$ be a generator of $\F_{3^2}^*$ satisfying $\zeta^2+2\zeta+2=0$. Choose $w=\zeta^{5}$, then
\begin{equation}\label{ex2-1}
G_w^1=\left(\begin{array}{llllllllllllll}
\zeta^6&   2& \zeta^2&   1& \zeta^6 &  2& \zeta^2&   1&   0 &  1 &  0 &  1 &  0 &\zeta^5\\
\zeta^7& \zeta^6 &\zeta^5 &  2 &\zeta^3& \zeta^2 &  \zeta &  1  & 0  & 0&   1&   1 &\zeta^5 &  0\\
 1 &  1   &1 &  1 &  1  & 1 &  1 &  1 &  1 &  0 &  0 &  0&   2&   1
\end{array}\right).
\end{equation}
By Magma \cite{Magma}, it's easy to check that the linear code $\cC_w^1$ generated by (\ref{ex2-1}) is a $[14,3,11]$ NMDS code with weight enumerator
$$
A(z)=1 + 160 z^{11} + 248 z^{12} +  144 z^{13} + 176 z^{14},
$$
which coincides with the conclusion of Theorem \ref{q+5odd}.
\end{example}

\begin{example}\label{ex3}
Let $q=11$. Choose $w=7$, then
\begin{equation}\label{ex3-1}
G_w^2=\left(\begin{array}{llllllllllllllll}
1&  4&  9&  5&  3 & 3 & 5 & 9&  4 & 1&  0&  1 & 0&  1 & 0&  7\\
10&  9 & 8 & 7 & 6 & 5&  4&  3  &2 & 1 & 0 & 0 & 1 & 1&  7&  0\\
1&  1 & 1 & 1 & 1 & 1 & 1 & 1 & 1 & 1&  1 & 0  &0 & 0& 10 & 1
\end{array}\right).
\end{equation}
By Magma \cite{Magma}, it's easy to check that the linear code $\cC_w^2$ generated by (\ref{ex3-1}) is a $[16,3,13]$ NMDS code with weight distribution
$$
A(z)=1 + 230 z^{13} + 510 z^{14} +  210 z^{15} + 380 z^{16},
$$
which coincides with the conclusion of Theorem \ref{q+5odd}.
\end{example}

\section{Optimal Locally Recoverable Codes}\label{Sec: LRC}

For locally recoverable codes (LRCs), a symbol is lost due to a node failure, its value can be recovered if every coordinate of the codeword $\vc\in \mathcal{C}$ can be recovered from a subset of $r$ other coordinates of $\vc$. Mathematically,
a code $\mathcal {C}$ has {\it locality $r$} if for every $i\in[n]:=\{1, 2, \cdots , n\}$ there exists a subset $R_i \subseteq [n]\setminus\{i\}$, $\# R_i\leq r$ and a function $\phi_i$ such that for every codeword $\vc=(c_1,\ldots,c_n)\in\mathcal C$,
$$c_i=\phi_i(\{c_j, j\in R_i\}).$$

An $(n,k,r)$-LRC code $\mathcal{C}$ over $\mathbb {F}_q$ is of code length $n$, cardinality $q^k$, and locality $r$. The parameters of an $(n,k,r)$-LRC code have been intensively studied.

\begin{lemma}[\cite{Gopalan2012}, Singleton-like bound]
Let $\mathcal {C}$ be an $(n,k,r)$-LRC code over $\mathbb {F}_q$, then
the minimum distance of $\mathcal {C}$ satisfies
\begin{equation}\label{Singleton}
d \leq n-k-\lceil\frac{k}{r}\rceil+2.
\end{equation}
\end{lemma}

Note that if $r=k$, the upper bound (\ref{Singleton}) coincides with the classical Singleton bound $d\leq n-k+1$.
LRC codes for which $d= n-k-\lceil k/r \rceil+2$ are called {\it $d$-optimal} LRC codes.

\begin{lemma}[\cite{Cadambe2013}, Cadambe-Mazumdar bound]
For any $(n, k, r)$-LRC code with minimum distance $d$ over $\mathbb {F}_q$, its dimension satisfies
\begin{equation}\label{dim-bound}
k \leq \min _{t \in \mathbb{Z}^{+}}\left[t r+k_{\mathrm{opt}}^{(q)}(n-t(r+1), d)\right],
\end{equation}
where $k_{\mathrm{opt}}^{(q)}(n, d)$ is the largest possible dimension of an $n$-length code, for a given alphabet
size $q$ and a given minimum distance $d$. $\mathbb{Z}^{+}$ represents the set of all positive integers.
\end{lemma}

LRC codes for which $k = \min \limits_{t \in \mathbb{Z}^{+}}\left[t r+k_{\mathrm{opt}}^{(q)}(n-t(r+1), d)\right]$ are called {\it $k$-optimal} LRC codes.

Let
$$\mathcal{B}_{k}\left(\cC^{\perp}\right)=\{supp(\textbf{c}^{\bot}): \textbf{c}^{\bot} \in \cC^{\perp}, wt(\textbf{c}^{\bot})=k\}.$$
The following lemma gives a way to calculate the locality of NMDS codes.
%This property is proved in \cite{Tan2021}, we still present a brief proof for completeness.

\begin{lemma}[\cite{Tan2021}]\label{localities}
Let $\mathcal {C}$ be an $[n,k,n-k]$ NMDS code. Then the code $\mathcal {C}$ has locality $ k-1$,  if
$
\bigcup\limits_{S \in \mathcal{B}_{k}\left(\mathcal{C}^{\perp}\right)} \mathcal{S}=[n].
$
The dual code $\mathcal {C}^{\perp}$ has locality $ n-k-1$,  if
$
\bigcap\limits_{S \in \mathcal{B}_{k}\left(\mathcal{C}^{\perp}\right)} \mathcal{S}=\emptyset.
$

\end{lemma}
%\proof
%Since $\bigcup_{S \in \mathcal{B}_{k}\left(\mathcal{C}^{\perp}\right)} \mathcal{S}=[n],$ for any $i \in [n],$ there exists a $k$-weight codeword $\textbf{c}^{\perp} \in \mathcal{C}^{\perp}$ such that $i \in supp(\textbf{c}^{\perp})$. By the definition of locality, we get $\mathcal {C}$ has locality $k-1$.
%
%Since $\bigcap_{S \in \mathcal{B}_{k}\left(\mathcal{C}^{\perp}\right)} \mathcal{S}=\emptyset,$ for any $i \in [n],$ there exists a $k$-weight codeword $\textbf{c}^{\perp} \in \mathcal{C}^{\perp}$ such that $i \notin supp(\textbf{c}^{\perp})$. Then by Lemma \ref{weight2}, there is a $(n-k)$-weight codeword $\textbf{c} \in \mathcal{C}$ such that $i \in supp(\textbf{c})$. Therefore, $\mathcal{C}^{\perp}$ has locality $n-k-1$.
%\qed
%
%Thus the localities of NMDS code $\mathcal {C}$ and its dual code $\mathcal {C}^{\perp}$ can be determined by the union and intersection of the support sets of the minimum weight codewords in $\mathcal {C}^{\perp}$, respectively.
%

For an $[n,k,n-k]$ NMDS code $\cC$ with generator matrix $G=[{\bf g}_1,{\bf g}_2,\ldots,{\bf g}_n]$. Let $S_G=\{{\bf g}_1,{\bf g}_2,\ldots,{\bf g}_n\}$. Lemma \ref{weight2} states a fact that if there exists a ${\bf u}\in\F_q^k$ satisfying $H_{\bf u}\cap S_{G}=\{{\bf g}_{i_1},{\bf g}_{i_2},\ldots,{\bf g}_{i_k}\}$. Then there exists a a $k$-weight codeword $\textbf{c}^{\bot} \in \cC ^{\bot}$ such that $supp(\textbf{c}^{\bot})=\{i_1,i_2,\ldots,i_k\}$.

%
% Let
%$$
%	\cU=\left\{{\bf 0}\neq{\bf u}\in\F_q^k:\#({\cal H}_{\bf u}\cap S_G)=k\right\}.
%$$
%It is easy to obtain from Lemma \ref{Weight-lem1} that for any ${\bf u} \in \cU$, there exists a $k$-weight codeword $\textbf{c}^{\bot} \in \cC ^{\bot}$ such that
%\begin{equation}\label{mini-plane}
%	supp(\textbf{c}^{\bot})=\{i : g_i \in {\cal H}_{\bf u}\cap S_G\}.
%\end{equation}

\begin{theorem}\label{locality-even}
The $[q+5, 3, q+2]$ NMDS code $\mathcal {C}_v$ generated by (\ref{generator}) has locality $2$, and its dual code $\mathcal {C}_v^{\perp}$ has locality $q+1$.
\end{theorem}
\proof
By Lemma \ref{localities}, we only have to prove
$
\bigcup\limits_{S \in \mathcal{B}_{3}\left(\mathcal{C}_v^{\perp}\right)} \mathcal{S}=[q+5]
~~\mbox{and}~
\bigcap\limits_{S \in \mathcal{B}_{3}\left(\mathcal{C}_v^{\perp}\right)} \mathcal{S}=\emptyset.
$
From the proof of Theorem \ref{q+5}, it's easy to see that there are three codewords $\{{\bf c}_{01}^{\bot},{\bf c}_{02}^{\bot},{\bf c}_{03}^{\bot}\}$ of ${\cal C}^{\perp}_v$ satisfying
\begin{align*}
supp({\bf c}_{01}^{\bot})&=	\{q+1,q+2,q+3\},\\
supp({\bf c}_{02}^{\bot})&=	\{q,q+2,q+4\},\\
supp({\bf c}_{03}^{\bot})&=	\{q,q+1,q+5\}.
\end{align*}
Obviously $\bigcap\limits_{i=1}^3 supp({\bf c}_{0i}^{\bot})=\emptyset$, which results in $\bigcap\limits_{S \in \mathcal{B}_{3}\left(\mathcal{C}_v^{\perp}\right)} \mathcal{S}=\emptyset$. Therefore the dual code ${\cal C}^{\perp}_v$  has locality $n-k-1=q+1$.

For any $\alpha_i\in\F_q^*$, define $u_3=f(\alpha_i)+\alpha_i$ and ${\bf u}=(1,1,u_3)$. By the proof of Theorem \ref{q+5}, we have ${\cal H}_{\bf u}\cap S_{G_v}=\{{\bf g}_i,{\bf g}_{i'},{\bf g}_{q+3}\}$, where $i'\in [q]$ is the unique number distinct with $i$ such that $f(\alpha_{i'})+\alpha_{i'}=u_3$. That is, for any $i \in [q-1]$, there exists a codeword ${\bf c}_i^{\bot}\in{\cal C}^{\bot}$ satisfying $supp({\bf c}_{i}^{\bot})=\{i,i',q+3\}$. Then
\[\left(\bigcup\limits_{i=0}^{q-1}supp({\bf c}_{i}^{\bot})\right)\cup \left(\bigcup\limits_{i=1}^3 supp({\bf c}_{0i}^{\bot})\right)=[q+5],
\]
which results in $\bigcup\limits_{S \in \mathcal{B}_{3}\left(\mathcal{C}_v^{\perp}\right)} \mathcal{S}=[q+5]$. Therefore the code ${\cal C}_v$ has locality $k-1=2$.
\qed

\begin{theorem}\label{d-optimal-even}
The NMDS code $\mathcal {C}_v$ generated by (\ref{generator}) and its dual code $\mathcal {C}_v^{\perp}$ are both $d$-optimal and $k$-optimal LRC codes.
\end{theorem}
\proof
Putting the parameters of $\mathcal {C}_v$ into the Singleton-like bound (\ref{Singleton}), we have
$$q+2=d \leq n-k-\lceil\frac{k}{r}\rceil+2=q+5-3-\lceil\frac{3}{2}\rceil+2=q+2.$$
Therefore, $\mathcal {C}$ is a $d$-optimal LRC code.

Putting the parameters of $\mathcal {C}_v$ into the Cadambe-Mazumdar bound (\ref{dim-bound}), we have
$$
\begin{aligned}
3=k & \leq \min _{t \in \mathbb{Z}^{+}}\left\{r t+k_{o p t}^{(q)}(n-(r+1) t, d)\right\} \\
& \leq r+k_{o p t}^{(q)}(n-(r+1), d) \\
&=2+k_{o p t}^{(q)}(q+5-(2+1), q+2) \\
&=3.
\end{aligned}
$$
where the last equality holds as $k_{o p t}^{(q)}(q+2, q+2)$=1 by the classical Singleton bound. Therefore, $\mathcal {C}$ is a $k$-optimal LRC code.

Similarly, we can prove that $\mathcal {C}_v^{\perp}$ is both $d$-optimal and $k$-optimal LRC code.
\qed

%Similar results can be obtained for the NMDS code $\cC_w$ generated by (\ref{generator2}), and the proof is omitted.

\begin{theorem}
The NMDS code $\mathcal {C}_w$ generated by (\ref{generator2}) has locality $2$, its dual code $\mathcal {C}^{\perp}$ has locality $q+1$. In addition, they are both  $d$-optimal and $k$-optimal LRC codes.
\end{theorem}
\proof
The proof is almost the same as that of Theorem \ref{locality-even} and Theorem \ref{d-optimal-even}, so we omit it here.
\qed

%\begin{theorem}
%The NMDS code $\mathcal {C}$ generated by (\ref{generator2}) and its dual code $\mathcal {C}^{\perp}$ are both the $d$-optimal and $k$-optimal LRC codes.
%\end{theorem}

\section{Conclusion}\label{Conclusion}

In this paper, we present two constructions of  $[q+5,3,q+2]$ NMDS codes over $\F_q$ for any prime power $q$, by applying hyperovals or ovals in PG$(2,q)$ respectively. The main idea is to add some suitable points of PG$(2,q)$ to the ovals or hyperovals, such that the resultant  point-sets are $(q+5,3)$-arcs in PG$(2,q)$, which are equivalent to $[q+5,3,q+2]$ NMDS codes over $\F_q$. We call this procedure as an extension of arcs.
The weight enumerators and the localities of these NMDS codes are also determined. It turns out that
our resultant NMDS codes and their dual codes are both distance-optimal and dimension-optimal locally repairable codes.

Let $m'(3,q)$ denote the maximal value of $n$ for an $[n,3,n-3]$ NMDS codes
to exist. It is conjectured that $m'(3,q)\leq 2q-1$ for $q\geq8$ in \cite{Dodunekov1995}, and this is only known to be true for $q=8,9$. Researches on $m'(3,q)$ are mostly focused on fixed small $q$, such as \cite{2014Bartoli}, \cite{1999Marcugini}, \cite{2002Marcugini} and the references therein.

By the extension of hyperovals in PG$(2,2^m)$, we can find by magma more examples of longer $[n,3,n-3]$  with $n\in\{2^3+7,2^4+9,2^5+11\}$.  Especially, the $[2^3+7,3,12]$ NMDS code over $\F_{2^3}$ can be obtained
by the generator matrix
 \begin{equation}
G=\left(\begin{array}{lllllllllllllll}
\varsigma^5& \varsigma^3 &  \varsigma& \varsigma^6 &\varsigma^4 &\varsigma^2 &  1 &  0 &  1 &  0&   1 &  0 &\varsigma^5 &  \varsigma &\varsigma^2\\
\varsigma^6& \varsigma^5& \varsigma^4& \varsigma^3& \varsigma^2&   \varsigma &  1  & 0 &  0 &  1  & 1& \varsigma^5 &  0& \varsigma^3&   1\\
  1 &  1 &  1 &  1 &  1  & 1 &  1 &  1  & 0 &  0 &  0  & 1  & 1 &  1 &  1
\end{array}\right),
\end{equation}
where $\varsigma$ is a generator of $\F_{2^3}^*$ satisfying $\varsigma^3+\varsigma+1=0$. It is worth noting that
the code length $n=15=2\times8-1$ coincides with the result of $m'(3,8)$, and also exceeds the maximum code length $n'=2^3+\lfloor 2\sqrt{2^3}\rfloor+1 =14$  obtained by elliptic curves.

An interesting open problem is whether there exist more
 NMDS codes whose code-lengths exceed that of the elliptic curve codes.
%An interesting open problem is whether there are more NMDS codes whose code-length exceeds that of the elliptic curve codes by the extension of arcs.
Another open problem is to determine the maximum extension number from arcs to $(n,3)$-arcs in PG$(2,q)$. It would be nice if this problem could be determined, even for some fixed classes of $q$.

%\section*{Acknowledgments}
%The authors are very grateful to the reviewer and the
%Associate Editor, Prof. James W.P. Hirschfeld, for their comments and
%suggestions that improved the presentation and quality of this
%paper. This work was finished when the authors visited the Hong Kong University of Science and Technology. The authors
%are grateful to Professor Cunsheng Ding for bringing them together
%in the summer of 2014.

%\section*{Acknowledgements}
%The authors are grateful to the anonymous reviewers for careful reading and for invaluable suggestions which improve the quality of the paper.

\end{document}